\documentclass[reqno]{amsart}%
\usepackage{amsfonts}
\usepackage{amsmath}
\usepackage{amssymb}
\usepackage{graphicx}%
\graphicspath{ {celebro/},{Image with noise/},{Images without noise/},{Lena with noise/}}
\newtheorem{theorem}{Theorem}
\theoremstyle{plain}

\newtheorem{definition}{Definition}

\newtheorem{lemma}{Lemma}
\newtheorem{notation}{Notation}

\newtheorem{proposition}{Proposition}
\newtheorem{remark}{Remark}

\numberwithin{equation}{section}
\begin{document}
\title[$p$-adic CNNs and Image Processing]{$p$-adic Cellular Neural Networks: Applications to Image Processing}
\author[Zambrano-Luna]{B. A. Zambrano-Luna}
\address{Centro de Investigaci\'{o}n y de Estudios Avanzados del Instituto
Polit\'{e}cnico Nacional\\
Departamento de Matem\'{a}ticas, Av. Instituto Polit\'{e}cnico Nacional \#
2508, Col. San Pedro Zacatenco, CDMX. CP 07360 \\
M\'{e}xico.}
\email{bazambrano@math.cinvestav.mx}
\author[Z\'{u}\~{n}iga-Galindo]{W. A. Z\'{u}\~{n}iga-Galindo}
\address{University of Texas Rio Grande Valley\\
School of Mathematical \& Statistical Sciences\\
One West University Blvd\\
Brownsville, TX 78520, United States }
\email{wilson.zunigagalindo@utrgv.edu}
\thanks{The second author was partially supported by the Lokenath Debnath Endowed Professorship.}

\begin{abstract}
The $p$-adic cellular neural networks (CNNs) are mathematical generalizations
of the neural networks introduced by Chua and Yang in the 80s. In this work we
present two new types of CNNs that can perform computations with real data,
and whose dynamics can be understood almost completely. The first type of
networks are edge detectors for grayscale images. The stationary states of
these networks are organized hierarchically in a lattice structure. The
dynamics of any of these networks consists of transitions toward some minimal
state in the lattice. The second type is a new class of reaction-diffusion
networks. We investigate the stability of these networks and show that they
can be used as filters to reduce noise, preserving the edges, in grayscale
images polluted with additive Gaussian noise. The networks introduced here
were found experimentally. They are abstract evolution equations on spaces of
real-valued functions defined in the $p$-adic unit ball for some prime number
$p$.\ In practical applications the prime $p$ is determined by the size of
image, and thus, only small primes are used. We provide several numerical
simulations\ showing how these networks\ work.

\end{abstract}
\maketitle

\section{Introduction}

In the late 80s Chua and Yang introduced a new natural computing paradigm
called the cellular neural networks (or cellular nonlinear networks) CNN which
includes the cellular automata as a particular case \cite{Chua-Yang},
\cite{Chua-Yang2}, \cite{Chua}. This paradigm has been extremely successful in
various applications in vision, robotics and remote sensing, see, e.g.,
\cite{Chua-Tamas}, \cite{Slavova} and the references therein.

In \cite{Zambrano-Zuniga} we introduce the $p$-adic cellular neural networks
which are mathematical generalizations of the classical CNNs. The new networks
have infinitely many cells which are organized hierarchically in rooted trees,
and also they have infinitely many hidden layers. Intuitively, the $p$-adic
CNNs occur as limits of large hierarchical discrete CNNs. A $p$-adic CNN is
the dynamical system given by%
\begin{equation}
\left\{
\begin{array}
[c]{l}%
\frac{\partial X(x,t)}{\partial t}=-X(x,t)+%
{\displaystyle\int\limits_{\mathbb{Q}_{p}}}
A(x,y)Y(y,t)dy+%
{\displaystyle\int\limits_{\mathbb{Q}_{p}}}
B(x,y)U(y)dy+Z(x)\\
\\
Y(x,t)=f(X(x,t)),
\end{array}
\right.  \label{Eq-1}%
\end{equation}
where $x$ is a $p$-adic number ($x\in\mathbb{Q}_{p}$), while $t$ is a
non-negative real number, $X(x,t)\in\mathbb{R}$ is the state of cell $x$ at
the time $t$, $Y(x,t)\in\mathbb{R}$ is the output of cell $x$ at the time $t$,
$f$ is a sigmoidal nonlinearity, $U$ is the input of the $\text{CNN}$, and $Z$
is the threshold of the $\text{CNN. In \cite{Zambrano-Zuniga}, }$we study the
Cauchy problem associated to (\ref{Eq-1}) and also provide numerical methods
for solving it.

The goal of this article is to show that $p$-adic CNNs can perform
computations using real data, and that the dynamics can be understood almost
completely. We present two new types of $p$-adic CNNs, one type for edge
detection of grayscale images, and the other, for denoising of grayscale
images polluted with Gaussian noise. It is important to emphasize that our
goal is not to produce new techniques for image processing, but to use these
tasks to verify that $p$-adic CNNs can perform relevant computations. On the
other hand, classical CNNs have been implemented in hardware for performing
certain image processing tasks. We have used some of the ideas introduced in
\cite{Chua-Tamas}, but our results go in a completely new direction.

We found experimentally that $p$-adic CNNs of the form
\begin{equation}
\left\{
\begin{array}
[c]{ll}%
\frac{\partial}{\partial t}X(x,t)=-X(x,t)+aY(x,t)+(B\ast U)(x)+Z(x), &
x\in\mathbb{Z}_{p},t\geq0;\\
& \\
Y(x,t)=f(X(x,t)), &
\end{array}
\right.  \label{Eq-2}%
\end{equation}
can be used as edge detectors, here $\mathbb{Z}_{p}$ is the $p$-adic unit
ball, and $U$ is an image. We develop numerical algorithms for solving the
Cauchy problem attached to (\ref{Eq-2}), with initial datum $X(x,0)=0$. The
simulations show that after a time sufficiently large the network outputs a
black-and-white image approximating the edges of the original image $U(x)$.
The performance of this edge detector is comparable to the Canny detector, and
other well-known detectors. But most importantly, we can explain, reasonably
well, how the network detects the edges of an image.

We determine all the stationary states of (\ref{Eq-2}), i.e. the solutions of
$\frac{\partial}{\partial t}X(x,t)=0$, for any $a\in\mathbb{R}$, see Lemma
\ref{Lemma-0} and Theorem \ref{Theorem1}. We show that for $a>1$, the set of
all possible stationary states $\mathcal{M}$ of (\ref{Eq-2}) has a
hierarchical structure, more precisely, $\left(  \mathcal{M},\preccurlyeq
\right)  $ is a lattice, where $\preccurlyeq$ is a partial order. Furthermore,
we determine the set of minimal elements of $\left(  \mathcal{M}%
,\preccurlyeq\right)  $, see Theorem \ref{Theorem2}. The dynamics of the
network consists of transitions in a hierarchically organized
landscape\textit{ }$\left(  \mathcal{M},\preccurlyeq\right)  $ toward some
minimal state. This is a reformulation of the classical paradigm asserting
that the dynamics of a large class of complex systems can be modeled as a
random walk on its energy landscape, see, e.g., \cite{KKZuniga}, \cite{Kozyrev
SV}.

We found experimentally that $p$-adic CNNs of the form%
\begin{gather}
\frac{\partial X(x,t)}{\partial t}=\mu X(x,t)+(\lambda I-\boldsymbol{D}%
_{0}^{\alpha})X(x,t)+%
{\displaystyle\int\limits_{\mathbb{Z}_{p}}}
A(x-y)f(X(y,t))dy\label{Eq-3}\\
+%
{\displaystyle\int\limits_{\mathbb{Z}_{p}}}
B(x-y)U(y)dy+Z(x)\nonumber
\end{gather}
can be used for denoising grayscale images polluted with Gaussian noise. In
this case, $X(x,0)$ is the input image, and $X(x,t_{0})$ is the output image,
for a suitable (typically small) $t_{0}$.

The CNN (\ref{Eq-3}) is a reaction-diffusion network. The diffusion part
corresponds to
\begin{equation}
\frac{\partial X(x,t)}{\partial t}=(\lambda I-\boldsymbol{D}_{0}^{\alpha
})X(x,t)\text{, }x\in\mathbb{Z}_{p},t\geq0, \label{Eq-4}%
\end{equation}
here $\boldsymbol{D}_{0}^{\alpha}$ is the Vladimirov operator acting on
functions supported in the unit ball, $\alpha>0$. The equation (\ref{Eq-4}) is
a $p$-adic heat equation in the unit ball, this means that there is a
stochastic Markov process attached to it. The paths of this stochastic process
are discontinuous. $p$-Adic heat equations and the associated stochastic
processes have been studied intensively in the last thirty years in connection
with models of complex systems, see, e.g., \cite{Av-4}-\cite{Av-5},
\cite{Dra-Kh-K-V}, \cite{KKZuniga}, \cite{Kochubei}-\cite{Kozyrev SV},
\cite{Torresblanca-Zuniga 1}-\cite{V-V-Z}, \cite{Zuniga-Nonlinearity}%
-\cite{Zuniga-LNM-2016}.

The reaction term in (\ref{Eq-3}) gives an estimation of the edges of the
image, while the diffusion term produces a smoothed version of the image.
Under suitable hypotheses, see Theorem \ref{Theorem3}, we show that a solution
of the initial value problem attached to (\ref{Eq-3}) is bounded at very time
if $\mu\leq0$, otherwise, the solution is bounded by $Ce^{\mu t}$, where $C$
is a positive constant. Some numerical simulations show that our filter
effectively reduces the noise while preserves the edges of the image, however,
its performance is inferior to the Perona-Malik filter, see, e.g.,
\cite{Sapiro}.

Finally, we want to mention that $p$-adic numbers have been used before in
processing image algorithms, see, e.g., \cite{Benois-Pineau et al
1}-\cite{Benois-Pineau et al 2}, \cite{KOtovich}. But these results are not
directly related with the ones presented here.

\section{ Basic facts on $p$-adic analysis}

In this section we fix the notation and collect some basic results about
$p$-adic analysis that we will use through the article. For a detailed
exposition on $p$-adic analysis the reader may consult \cite{Alberio et al},
\cite{Taibleson}, \cite{V-V-Z}. For a quick review of $p$-adic analysis the
reader may consult \cite{Bocardo-Zuniga-2}, \cite{Leon-Zuniga}.

\subsection{The field of $p$-adic numbers}

Throughout this article $p$ will denote a prime number. The field of $p-$adic
numbers $\mathbb{Q}_{p}$ is defined as the completion of the field of rational
numbers $\mathbb{Q}$ with respect to the $p-$adic norm $|\cdot|_{p}$, which is
defined as
\[
|x|_{p}=%
\begin{cases}
0 & \text{if }x=0\\
p^{-\gamma} & \text{if }x=p^{\gamma}\dfrac{a}{b},
\end{cases}
\]
where $a$ and $b$ are integers coprime with $p$. The integer $\gamma
=ord_{p}(x)$ with $ord_{p}(0):=+\infty$, is called the\textit{\ }%
$p-$\textit{adic order of} $x$. The metric space $\left(  \mathbb{Q}%
_{p},\left\vert \cdot\right\vert _{p}\right)  $ is a complete ultrametric
space. Ultrametric means that $\left\vert x+y\right\vert _{p}\leq\max\left\{
\left\vert x\right\vert _{p},\left\vert y\right\vert _{p}\right\}  $. As a
topological space $\mathbb{Q}_{p}$\ is homeomorphic to a Cantor-like subset of
the real line, see, e.g., \cite{Alberio et al}, \cite{V-V-Z}.

Any $p-$adic number $x\neq0$ has a unique expansion of the form
\begin{equation}
x=p^{ord_{p}(x)}\sum_{j=0}^{\infty}x_{j}p^{j}, \label{expansion}%
\end{equation}
where $x_{j}\in\{0,1,2,\dots,p-1\}$ and $x_{0}\neq0$. It follows from
(\ref{expansion}), that any $x\in\mathbb{Q}_{p}\smallsetminus\left\{
0\right\}  $ can be represented uniquely as $x=p^{ord_{p}(x)}u\left(
x\right)  $ and $\left\vert x\right\vert _{p}=p^{-ord_{p}(x)}$.

\subsection{Topology of $\mathbb{Q}_{p}$}

For $r\in\mathbb{Z}$, denote by $B_{r}(a)=\{x\in\mathbb{Q}_{p};\left\vert
x-a\right\vert _{p}\leq p^{r}\}$ \textit{the ball of radius }$p^{r}$
\textit{with center at} $a\in\mathbb{Q}_{p}$, and take $B_{r}(0):=B_{r}$. The
ball $B_{0}$ equals \textit{the ring of }$p-$\textit{adic integers
}$\mathbb{Z}_{p}$. We also denote by $S_{r}(a)=\{x\in\mathbb{Q}_{p};\left\vert
x-a\right\vert _{p}=p^{r}\}$ \textit{the sphere of radius }$p^{r}$
\textit{with center at} $a\in\mathbb{Q}_{p}$, and take $S_{r}(0):=S_{r}$. We
notice that $S_{0}=\mathbb{Z}_{p}^{\times}$ (the group of units of
$\mathbb{Z}_{p}$). The balls and spheres are both open and closed subsets in
$\mathbb{Q}_{p}$. In addition, two balls in $\mathbb{Q}_{p}$ are either
disjoint or one is contained in the other.

As a topological space $\left(  \mathbb{Q}_{p},\left\vert \cdot\right\vert
_{p}\right)  $ is totally disconnected, i.e. the only connected \ subsets of
$\mathbb{Q}_{p}$ are the empty set and the points. A subset of $\mathbb{Q}%
_{p}$ is compact if and only if it is closed and bounded in $\mathbb{Q}_{p}$,
see e.g. \cite[Section 1.3]{V-V-Z}, or \cite[Section 1.8]{Alberio et al}. The
balls and spheres are compact subsets. Thus $\left(  \mathbb{Q}_{p},\left\vert
\cdot\right\vert _{p}\right)  $ is a locally compact topological space.

Since $(\mathbb{Q}_{p},+)$ is a locally compact topological group, there
exists a Haar measure $dx$, which is invariant under translations, i.e.
$d(x+a)=dx$. If we normalize this measure by the condition $\int
_{\mathbb{Z}_{p}}dx=1$, then $dx$ is unique. For a quick review of the
integration in the $p$-adic framework the reader may consult
\cite{Bocardo-Zuniga-2}, \cite{Leon-Zuniga} and the references therein.

\begin{notation}
We will use $\Omega\left(  p^{-r}\left\vert x-a\right\vert _{p}\right)  $ to
denote the characteristic function of the ball $B_{r}(a)$.
\end{notation}

\subsection{The Bruhat-Schwartz space}

A real-valued function $\varphi$ defined on $\mathbb{Q}_{p}$ is \textit{called
locally constant} if for any $x\in\mathbb{Q}_{p}$ there exist an integer
$l(x)\in\mathbb{Z}$ such that%
\begin{equation}
\varphi(x+x^{\prime})=\varphi(x)\text{ for any }x^{\prime}\in B_{l(x)}.
\label{local_constancy}%
\end{equation}
A function $\varphi:\mathbb{Q}_{p}\rightarrow\mathbb{C}$ is called a
\textit{Bruhat-Schwartz function (or a test function)} if it is locally
constant with compact support. Any test function can be represented as a
linear combination, with real coefficients, of characteristic functions of
balls. The $\mathbb{R}$-vector space of Bruhat-Schwartz functions is denoted
by $\mathcal{D}(\mathbb{Q}_{p})$. For $\varphi\in\mathcal{D}(\mathbb{Q}_{p})$,
the largest number $l=l(\varphi)$ satisfying (\ref{local_constancy}) is called
\textit{the exponent of local constancy (or the parameter of constancy) of}
$\varphi$. Let $U$ be an open subset of $%
\mathbb{Q}
_{p}$, we denote by $\mathcal{D}(U)$ the $\mathbb{R}$-vector space of all test
functions with support in $U$. For instance $\mathcal{D}(\mathbb{Z}_{p})$ is
the $\mathbb{R}$-vector space of all test functions with supported in $\ $the
unit ball $\mathbb{Z}_{p}$. A function $\varphi$ in $\mathcal{D}%
(\mathbb{Z}_{p})$ can be written as%
\[
\varphi\left(  x\right)  =%
{\displaystyle\sum\limits_{j=1}^{M}}
\varphi\left(  \widetilde{x}_{j}\right)  \Omega\left(  p^{r_{j}}\left\vert
x-\widetilde{x}_{j}\right\vert _{p}\right)  ,
\]
where the $\widetilde{x}_{j}$, $j=1,\ldots,M$, are points in $\mathbb{Z}_{p}$,
the $r_{j}$, $j=1,\ldots,M$, are integers, and $\Omega\left(  p^{r_{j}%
}\left\vert x-\widetilde{x}_{j}\right\vert _{p}\right)  $ denotes the
characteristic function of the ball $B_{-r_{j}}(\widetilde{x}_{j}%
)=\widetilde{x}_{j}+p^{r_{j}}\mathbb{Z}_{p}$.

\subsection{Some function spaces}

Given $\rho\in\lbrack1,\infty)$, we denote by $L^{\rho}:=L^{\rho}\left(
\mathbb{Z}_{p}\right)  :=L^{\rho}\left(  \mathbb{Z}_{p},dx\right)  ,$ the
$\mathbb{R}-$vector space of all functions $g:\mathbb{Z}_{p}\rightarrow$
$\mathbb{R}$ satisfying%
\[
\left\Vert g\right\Vert _{\rho}=\left[
{\textstyle\int\limits_{\mathbb{Z}_{p}}}
\left\vert g\left(  x\right)  \right\vert ^{\rho}dx\right]  ^{\frac{1}{\rho}%
}<\infty.
\]
We denote by $\mathcal{C}(\mathbb{Z}_{p})$ the $\mathbb{R}$-vector space of
continuous functions $f:\mathbb{Z}_{p}\rightarrow\mathbb{R}$ satisfying
\begin{equation}
\left\Vert f\right\Vert _{\infty}:=\max_{x\in\mathbb{Z}_{p}}\text{ }\left\vert
f\left(  x\right)  \right\vert <\infty\text{.} \label{Condition_L_inf}%
\end{equation}

\section{$p$-Adic continuous CNNs}

\subsection{A type $p$-adic continuous CNNs}

In this section we present new edge detectors based on $p$-adic CNNs for
grayscale images. We take $B\in L^{1}(\mathbb{Z}_{p})$ and $U,Z\in
\mathcal{C}(\mathbb{Z}_{p})$, $a$, $b\in\mathbb{R}$ , and fix the sigmoidal
function $f(s)=\frac{1}{2}(\left\vert s+1\right\vert -|s-1|)$ for
$s\in\mathbb{R}$. In this section we consider the following $p$-adic CNN:%
\begin{equation}
\left\{
\begin{array}
[c]{ll}%
\frac{\partial}{\partial t}X(x,t)=-X(x,t)+aY(x,t)+(B\ast U)(x)+Z(x), &
x\in\mathbb{Z}_{p},t\geq0;\\
& \\
Y(x,t)=f(X(x,t)). &
\end{array}
\right.  \label{CNN_1}%
\end{equation}
We denote this $p$-adic CNN as $CNN(a,B,U,Z)$, where $a,B,U,Z$\ are the
parameters of the network. In applications to edge detection, we take $U(x)$
to be a grayscale image, and take the initial datum as $X(x,0)=0$.

\subsection{Stationary states}

We say that $X_{stat}(x)$ is a \textit{stationary state} of the network
$CNN(a,B,U,Z)$, if
\begin{equation}
\left\{
\begin{array}
[c]{ll}%
X_{stat}(x)=aY_{stat}(x)+(B\ast U)(x)+Z(x), & x\in\mathbb{Z}_{p};\\
& \\
Y_{stat}(x)=f(X_{stat}(x)). &
\end{array}
\right.  \label{Sationary state}%
\end{equation}

\begin{remark}
\label{Key-Remark}Let $\tilde{X}(x)$ be any solution of (\ref{Sationary state}%
). Then%
\begin{equation}
\widetilde{X}(x)=\left\{
\begin{array}
[c]{lcc}%
a+(B\ast U)(x)+Z(x) & \text{if} & \widetilde{X}(x)>1\\
-a+(B\ast U)(x)+Z(x) & \text{if} & \widetilde{X}(x)<-1,
\end{array}
\right.  \label{Stationary_Sol_1}%
\end{equation}
and
\begin{equation}
\left(  1-a\right)  \widetilde{X}(x)=(B\ast U)(x)+Z(x)\text{ if }\left\vert
\widetilde{X}(x)\right\vert \leq1. \label{Stationary_Sol_2}%
\end{equation}

\end{remark}

\begin{lemma}
\label{Lemma-0}(i) If $a<1$, then the network $CNN(a,B,U,Z)$ has a unique
stationary state $X_{stat}(x)\in\mathcal{C}(\mathbb{Z}_{p})$ given by%
\begin{equation}
X_{stat}(x)=\left\{
\begin{array}
[c]{lcr}%
a+(B\ast U)(x)+Z(x) & \text{if} & (B\ast U)(x)+Z(x)>1-a\\
-a+(B\ast U)(x)+Z(x) & \text{if} & (B\ast U)(x)+Z(x)<-1+a\\
\frac{(B\ast U)(x)+Z(x)}{1-a} & \text{if} & |(B\ast U)(x)+Z(x)|\leq1-a.
\end{array}
\right.  \label{Case_1}%
\end{equation}

\noindent(ii) If $a=1$ , then the network $CNN(a,B,U,Z)$ has a unique
stationary state $X_{stat}(x)\in L^{1}(\mathbb{Z}_{p})$ given by
\begin{equation}
X_{stat}(x)=\left\{
\begin{array}
[c]{lcr}%
1+(B\ast U)(x)+Z(x) & \text{if} & (B\ast U)(x)+Z(x)>0\\
-1+(B\ast U)(x)+Z(x) & \text{if} & (B\ast U)(x)+Z(x)<0\\
0 & \text{if} & (B\ast U)(x)+Z(x)=0.
\end{array}
\right.  \label{Case_2}%
\end{equation}

\end{lemma}

\begin{proof}
If $a<1$, it follows from (\ref{Stationary_Sol_1})-(\ref{Stationary_Sol_2})
that (\ref{Case_1}) is a continuous stationary state since by the dominated
convergence theorem $(B\ast U)(x)$ is continuous. To establish the uniqueness
of the solution, let $X(x)\in\mathcal{C}(\mathbb{Z}_{p})$ be another
stationary state. Consider a point $x_{0}\in\mathbb{Z}_{p}$ such that
$X(x_{0})>1$. Then by (\ref{Stationary_Sol_1}), $X(x_{0})=a+(B\ast
U)(x_{0})+Z(x_{0})>1$ consequently $(B\ast U)(x_{0})+Z(x_{0})>1-a$ and
therefore
\[
X(x_{0})=a+(B\ast U)(x_{0})+Z(x_{0}):=X_{stat}(x_{0}).
\]
The cases $X(x_{0})<-1$ and $X\left\vert (x_{0})\right\vert <1$ are treated in
a similar way.

The case $a=1$ follows from (\ref{Stationary_Sol_2}), in this case we have
that $X_{stat}(x)\in L^{1}(\mathbb{Z}_{p})$ since $X_{stat}(x)$\ is bounded.
The continuity of $X_{stat}(x)$ requires further hypotheses on $B,U,Z$.
\end{proof}

\begin{definition}
\label{Definition1}Assume that $a>1$. Given
\[
I_{+}\subseteq\{x\in\mathbb{Z}_{p};\;1-a<(B\ast U)(x)+Z(x)\},
\]%
\[
I_{-}\subseteq\{x\in\mathbb{Z}_{p};\;(B\ast U)(x)+Z(x)<a-1\},
\]
satisfying $I_{+}\cap I_{-}=\varnothing$ \ and
\[
\mathbb{Z}_{p}\smallsetminus\left(  I_{+}\cup I_{-}\right)  \subseteq
\{x\in\mathbb{Z}_{p};\;1-a<(B\ast U)(x)+Z(x)<a-1\},
\]
we define the function%
\begin{equation}
X_{stat}(x;I_{+},I_{-})=\left\{
\begin{array}
[c]{lll}%
a+(B\ast U)(x)+Z(x) & \text{if} & x\in I_{+}\\
-a+(B\ast U)(x)+Z(x) & \text{if} & x\in I_{-}\\
\frac{(B\ast U)(x)+Z(x)}{1-a} & \text{if} & x\in\mathbb{Z}_{p}\setminus\left(
I_{+}\cup I_{-}\right)  .
\end{array}
\right.  \label{Stationary_Sol_3}%
\end{equation}

\end{definition}

\begin{theorem}
\label{Theorem1}Assume that $a>1$. All functions of type
(\ref{Stationary_Sol_3}) are stationary states of the network $CNN(a,B,U,Z)$.
Conversely, any \ stationary state of the network $CNN(a,B,U,Z)$ has the form
(\ref{Stationary_Sol_3}).
\end{theorem}

\begin{proof}
We first verify that any function of type (\ref{Stationary_Sol_3}) is
a\ stationary state. Take a point $x_{0}\in\mathbb{Z}_{p}$. Since the sets
$I_{+}$, $I_{-}$, $\mathbb{Z}_{p}\smallsetminus\left(  I_{+}\cup I_{-}\right)
$ are disjoint, three cases occur.

\textbf{Case 1}: $x_{0}\in I_{+}$.

If $x_{0}\in I_{+}$, then $X_{stat}(x_{0};I_{+},I_{-})=a+(B\ast U)(x_{0}%
)+Z(x_{0})$ and by definition of $I_{+}$, $a+(B\ast U)(x_{0})+Z(x_{0})>1$.
Then
\begin{multline*}
af(X_{stat}(x_{0};I_{+},I_{-}))+(B\ast U)(x_{0})+Z(x_{0})=\\
a+(B\ast U)(x_{0})+Z(x_{0})=X_{stat}(x_{0};I_{+},I_{-}).
\end{multline*}

\textbf{Case 2}: $x_{0}\in I_{-}$.

If $x_{0}\in I_{-}$, then $X_{stat}(x_{0};I_{+},I_{-})=-a+(B\ast
U)(x_{0})+Z(x_{0})$ and by definition of $I_{-}$, $-a+(B\ast U)(x_{0}%
)+Z(x_{0})<-1$. Then
\begin{multline*}
af(X_{stat}(x_{0};I_{+},I_{-}))+(B\ast U)(x_{0})+Z(x_{0})=\\
-a+(B\ast U)(x_{0})+Z(x_{0})=X_{stat}(x_{0};I_{+},I_{-}).
\end{multline*}

\textbf{Case 3}: $x_{0}\in\mathbb{Z}_{p}\smallsetminus\left(  I_{+}\cup
I_{-}\right)  $.

If $x_{0}\notin I_{+}\sqcup I_{-}$, then $X_{stat}(x_{0};I_{+},I_{-}%
)=\frac{(B\ast U)(x_{0})+Z(x_{0})}{1-a}$ and by definition of $\mathbb{Z}%
_{p}\smallsetminus\left(  I_{+}\cup I_{-}\right)  $, $-1\leq\frac{(B\ast
U)(x_{0})+Z(x_{0})}{1-a}\leq1$, then%
\begin{multline*}
af(X_{stat}(x_{0};I_{+},I_{-}))+(B\ast U)(x_{0})+Z(x_{0})=\\
a\frac{(B\ast U)(x_{0})+Z(x_{0})}{1-a}+(B\ast U)(x_{0})+Z(x_{0})=\frac{(B\ast
U)(x_{0})+Z(x_{0})}{1-a}\\
=X_{stat}(x_{0};I_{+},I_{-}).
\end{multline*}

Therefore $X_{stat}(x_{0};I_{+},I_{-})$ is a stationary state of the network
$CNN(a,B,U,Z)$.

Now, suppose that $X_{stat}(x)$ is a stationary state of network
$CNN(a,B,U,Z)$. Set
\[
I_{+}:=X_{stat}^{-1}((1,\infty))\text{, \ \ \ \ }I_{-}:=X_{stat}^{-1}%
((-\infty,-1)).
\]
By using (\ref{Stationary_Sol_1}) and (\ref{Stationary_Sol_2}), we have
\[
I_{+}\subseteq(B\ast U+Z)^{-1}((1-a,\infty)),
\]%
\[
I_{-}\subseteq(B\ast U+Z)^{-1}((-\infty,a-1)),
\]
and
\[
\mathbb{Z}_{p}\smallsetminus\left(  I_{+}\cup I_{-}\right)  \subseteq(B\ast
U+Z)^{-1}([1-a,a-1]).
\]
Then $X_{stat}(x;I_{+},I_{-})$ is a well-defined function. Finally, using
again (\ref{Stationary_Sol_1}) and (\ref{Stationary_Sol_2}), we conclude that
$X_{stat}(x)=X_{stat}(x;I_{+},I_{-})$ for all $x\in\mathbb{Z}_{p}$.
\end{proof}

\begin{remark}
Notice that%
\[
Y_{stat}(x;I_{+},I_{-}):=f\left(  X_{stat}(x;I_{+},I_{-})\right)  =\left\{
\begin{array}
[c]{lll}%
1 & \text{if} & x\in I_{+}\\
-1 & \text{if} & x\in I_{-}\\
\frac{(B\ast U)(x)+Z(x)}{1-a} & \text{if} & x\in\mathbb{Z}_{p}\setminus\left(
I_{+}\cup I_{-}\right)  .
\end{array}
\right.
\]
The function $Y_{stat}(x;I_{+},I_{-})$ is the output of the network. If
$I_{+}\cup I_{-}=\mathbb{Z}_{p}$, we say that $X_{stat}(x;I_{+},I_{-})$ is
bistable. The set $\mathcal{B}\left(  I_{+},I_{-}\right)  =\mathbb{Z}%
_{p}\setminus\left(  I_{+}\cup I_{-}\right)  $ measures how far $X_{stat}%
(x;I_{+},I_{-})$ is from being bistable. We call set $\mathcal{B}\left(
I_{+},I_{-}\right)  $ the set of bistability of $X_{stat}(x;I_{+},I_{-})$. If
$\mathcal{B}\left(  I_{+},I_{-}\right)  =\varnothing$, then $X_{stat}%
(x;I_{+},I_{-})$ is bistable.
\end{remark}

\begin{remark}
If $I_{+}\cup I_{-}\subsetneqq\mathbb{Z}_{p}$, we say that $X_{stat}%
(x;I_{+},I_{-})$ is an unstable.
\end{remark}

\section{Hierarchical structure of the space of stationary states}

A relation $\preccurlyeq$ is \textit{a partial order} on a set $S$ if it
satisfies: 1 (reflexivity) $f\preccurlyeq f$ for all $f$ in $S$; 2
(antisymmetry) $f\preccurlyeq g$ and $g\preccurlyeq f$ implies $f=g$; 3
(transitivity) $f\preccurlyeq g$ and $g\preccurlyeq h$ implies $f\preccurlyeq
h$. \ A \textit{partially ordered set} $\left(  S,\preccurlyeq\right)  $ (or
poset) is a set endowed with a partial order. A partially ordered set $\left(
S,\preccurlyeq\right)  $ is called a \textit{lattice} if for every $f$, $g$ in
$S$, the elements $f\wedge g=\inf\{f,g\}$ and $f\vee$ $g=\sup\{f,g\}$ exist.
Here, $f\wedge g$ denotes the smallest element in $S$ satisfying $f\wedge
g\preccurlyeq f$ and $f\wedge g\preccurlyeq g$; while $f\vee$ $g$ \ denotes
the largest element in $S$ satisfying $f\preccurlyeq$ $f\vee$ $g$ and
$g\preccurlyeq f\vee$ $g$. We say that $h\in S$ a \textit{minimal} element of
with respect to $\preccurlyeq$, if there is no element $f\in S$, $f\neq h$
such that $f\preccurlyeq h$.

Posets offer a natural way to formalize the notion of hierarchy.

We set
\[
\mathcal{M}=\bigcup_{I_{+},I_{-}}\left\{  X_{stat}(x;I_{+},I_{-})\right\}  ,
\]
where $I_{+},I_{-}$ run trough all the sets given in Definition
\ref{Definition1}. \ Given $X_{stat}(x;I_{+},I_{-})$ and $X_{stat}%
(x;I_{+}^{\prime},I_{-}^{\prime})$ in $\mathcal{M}$, with $I_{+}\cup I_{-}%
\neq\mathbb{Z}_{p}$ or $I_{+}^{\prime}\cup I_{-}^{\prime}\neq\mathbb{Z}_{p}$,
we define%
\begin{equation}
X_{stat}(x;I_{+}^{\prime},I_{-}^{\prime})\preccurlyeq X_{stat}(x;I_{+}%
,I_{-})\text{ if }I_{+}\cup I_{-}\subseteq I_{+}^{\prime}\cup I_{-}^{\prime}.
\label{Definitioon_Order}%
\end{equation}
In the case $I_{+}\cup I_{-}=\mathbb{Z}_{p}$ and $I_{+}^{\prime}\cup
I_{-}^{\prime}=\mathbb{Z}_{p}$, the corresponding stationary states
$X_{stat}(x;I_{+},I_{-})$, $X_{stat}(x;I_{+},I_{-})$ are not comparable. Since
the condition $I_{+}\cup I_{-}\subseteq I_{+}^{\prime}\cup I_{-}^{\prime}\ $is
equivalent to $\mathcal{B}\left(  I_{+}^{\prime},I_{-}^{\prime}\right)
=\mathbb{Z}_{p}\setminus(I_{1}^{\prime}\cup I_{-1}^{\prime})\subseteq
\mathcal{B}\left(  I_{1},I_{-1}\right)  \mathcal{=}\mathbb{Z}_{p}%
\setminus(I_{1}\sqcup I_{-1})$, condition (\ref{Definitioon_Order}) means that
the set of bistability of $X_{stat}(x;I_{+}^{\prime},I_{-}^{\prime})$ is
smaller that the set of of bistability of $X_{stat}(x;I_{+},I_{-})$. Also, the
condition $I_{+}\cup I_{-}\subseteq I_{+}^{\prime}\cup I_{-}^{\prime}%
$implies\ that
\[
X_{stat}(x;I_{+}^{\prime},I_{-}^{\prime})(x)=X_{stat}(x;I_{+},I_{-})(x)\text{
for all }x\in I_{+}\cup I_{-}\cup\mathcal{B}\left(  I_{+}^{\prime}\cup
I_{-}^{\prime}\right)  .
\]
By using this observation, one verifies that (\ref{Definitioon_Order}) defines
a partial order in $\mathcal{M}$. This means that the set of stationary states
of the network $CNN(a,B,U,Z)$, $a>1$, has a hierarchical structure, where the
bistable stationary states are the minimal ones. Intuitively, the bistable
stationary states are at the deepest level of $\mathcal{M}$. Furthermore,
$\left(  \mathcal{M},\preccurlyeq\right)  $ is a lattice. Indeed, given
$X_{stat}(x;I_{+}^{\prime},I_{-}^{\prime})$, $X_{stat}(x;I_{+},I_{-})$,in
$\mathcal{M}$, it verifies that
\[
X_{stat}(x;I_{+}^{\prime},I_{-}^{\prime})\wedge X_{stat}(x;I_{+}%
,I_{-})=X_{stat}(x;I_{+}^{\prime\prime},I_{-}^{\prime\prime})\text{, }%
\]
where $I_{+}^{\prime\prime}=I_{+}\cup I_{+}^{\prime}$, $I_{-}^{\prime\prime
}=I_{-}\cup I_{-}^{\prime}$, and
\[
X_{stat}(x;I_{+}^{\prime},I_{-}^{\prime})\vee X_{stat}(x;I_{+},I_{-}%
)=X_{stat}(x;I_{+}^{\prime\prime\prime},I_{-}^{\prime\prime\prime}),
\]
where $I_{+}^{\prime\prime\prime}=I_{+}\cap I_{+}^{\prime}$, $I_{-}%
^{\prime\prime\prime}=I_{-}\cap I_{-}^{\prime}$. Therefore, we have
established the following result:

\begin{theorem}
\label{Theorem2}$\left(  \mathcal{M},\preccurlyeq\right)  $ is a lattice.
Furthermore, the set of minimal elements of $\left(  \mathcal{M}%
,\preccurlyeq\right)  $ agrees with the set of bistable states of
$CNN(a,B,U,Z)$.
\end{theorem}

\section{Edge detection}

\subsection{A new class of edge detectors}

We take $a>1$, $X(x,0)=0$, and $U(x)\in\mathcal{D}(\mathbb{Z}_{p})$ to be a
grayscale image. We argue that network (\ref{CNN_1}) works as an edge
detector. By Theorem \ref{Theorem1}, network $CNN(a,B,U,Z)$, $a>1$ has steady
states of the form
\begin{equation}
Y_{stat}(x)=f(X_{stat}\left(  x\right)  )=\left\{
\begin{array}
[c]{lll}%
+1 & \text{if} & (B\ast U)(x)+Z(x)>\text{Threshold}_{1}\\
-1 & \text{if} & (B\ast U)(x)+Z(x)<\text{Threshold}_{2},
\end{array}
\right.  \label{Steady State}%
\end{equation}
where Threshold$_{2}$, Threshold$_{1}$ are real numbers. This type of outputs
occur for networks with stationary states where $I_{+}\sqcup I_{-}%
=\mathbb{Z}_{p}$. For instance, when $I_{+}\subseteq(B\ast U+Z)^{-1}%
((Threshold_{1},\infty))$ and $I_{-}\subseteq(B\ast U+Z)^{-1}((-\infty
,Threshold_{2}))$. If $U(x)$ is sufficiently small, then $(B\ast U)(x)+Z(x)$
gives a measure of dispersion of the image intensities; if this value is
larger than Threshold$_{1}$, the networks outputs $+1$ to indicate the
existence of an edge, if value is smaller than Threshold$_{2}$, the network
outputs $-1$ to indicate the nonexistence of an edge.

We conducted several numerical experiments with grayscale images. We
implemented a numerical method for solving the initial value problem attached
to network $CNN(a,B,U,Z)$, with $X(x,0)=0$ and $U(x)$ a grayscale image. The
simulations show that after a sufficiently large time the network outputs a
black-and-white image approximating the edges of the original image $U(x)$.
This means that for $t$ sufficiently large $X(x,t)$ is close to a bistable
stationary state $X_{stat}(x;I_{+},I_{-})$. Furthermore, after a certain
sufficiently large time, the output of the network do not show a difference
perceivable by the human eye. We interpret this result as the bistable
stationary states are asymptotically stable; of course this is a mathematical conjecture.

We now give an intuitive picture of the dynamics of the network, for $t$
sufficiently large, using $\left(  \mathcal{M},\preccurlyeq\right)  $ as an
\textit{asymptotic} \textit{landscape for} $CNN(a,B,U,Z)$. For $t$
sufficiently large, the network performs transitions between stationary states
$X_{stat}(x;I_{+},I_{-})$ belonging to a small neighborhood $\mathcal{N}$
around a bistable state $X_{stat}^{\left(  0\right)  }(x;I_{+},I_{-})$, with
$I_{+}\cup I_{-}=\mathbb{Z}_{p}$. The dynamics of the network consists of
transitions in a hierarchically organized landscape\textit{ }$\left(
\mathcal{M},\preccurlyeq\right)  $ toward some minimal state. This is a
reformulation of the classical paradigm asserting that the dynamics of a large
class of complex systems can be modeled as a random walk on its energy landscape.

\subsection{Discretization}

To process an image $U(x)$, we use a discrete version of network
$CNN(a,B,U,Z)$, $a>1$. In turn, this requires to determine suitable kernels
$B(x)$. We address these matters on this section.

We take $L$ \ to be a positive integer, and set $G_{L}=\mathbb{Z}_{p}%
/p^{L}\mathbb{Z}_{p}$. We identify $i\in G_{L}$ with an element of the form%
\[
i=i_{0}+i_{1}p+\ldots+i_{L-1}p^{L-1},
\]
where the $i_{k}$s belong to the set $\left\{  0,1,\ldots,p-1\right\}  $. We
denote by $\mathcal{D}_{L}\left(  \mathbb{Z}_{p}\right)  $ the $\mathbb{R}%
$-vector space of test functions of the form%
\[
\varphi\left(  x\right)  =%
{\displaystyle\sum\limits_{i\in G_{L}}}
\varphi\left(  i\right)  \Omega\left(  p^{L}\left\vert x-i\right\vert
_{p}\right)
\]
supported in the unit ball $\mathbb{Z}_{p}$. Since $\Omega\left(
p^{L}\left\vert x-i\right\vert _{p}\right)  \Omega\left(  p^{L}\left\vert
x-j\right\vert _{p}\right)  =0$ for $i\neq j$, the set
\[
\left\{  \Omega\left(  p^{L}\left\vert x-i\right\vert _{p}\right)  \right\}
_{i\in G_{L}}%
\]
is a basis of $\mathcal{D}_{L}\left(  \mathbb{Z}_{p}\right)  $. Notice that
the dimension of $\mathcal{D}_{L}\left(  \mathbb{Z}_{p}\right)  $\ is $p^{L}$.

Assuming that $B\left(  x\right)  ,U(x),Z(x)\in\mathcal{D}_{L}\left(
\mathbb{Z}_{p}\right)  $, the initial value problem%
\begin{equation}
\left\{
\begin{array}
[c]{ll}%
\frac{\partial}{\partial t}X(x,t)=-X(x,t)+aY(x,t)+(B\ast U)(x)+Z(x), &
x\in\mathbb{Z}_{p},t\geq0;\\
& \\
X(x,0)=X_{0}\in\mathcal{D}_{L}\left(  \mathbb{Z}_{p}\right)  . &
\end{array}
\right.  \label{Cauchy_1}%
\end{equation}
has unique solution
\begin{equation}
X(x,t)=\sum_{i\in G_{L}}X(i,t)\Omega\left(  p^{L}\left\vert x-i\right\vert
_{p}\right)  \label{Definition_A}%
\end{equation}
in $\mathcal{D}_{L}\left(  \mathbb{Z}_{p}\right)  $ for $t\geq0$, see
\cite[Theorem 1]{Zambrano-Zuniga}.

This result allow us to obtain a discretization of (\ref{Cauchy_1}) and
(\ref{CNN_1}) as follows. Take%
\begin{equation}
U(x)=\sum_{i\in G_{L}}U(i)\Omega(p^{L}|x-i|), \label{Definition_AA}%
\end{equation}%
\begin{equation}
Z(x)=\sum_{i\in G_{L}}Z(i)\Omega(p^{L}|x-i|), \label{Definition_BB}%
\end{equation}
and%
\begin{equation}
B(x)=p^{M_{2}-M_{1}}\Omega(p^{M_{2}}|x|_{p})-\Omega(p^{M_{1}}|x|_{p})
\label{Definition_B}%
\end{equation}
for some integers $M_{1}\leq M_{2}\leq L$. \ 

We now take $i$, $j\in G_{L}$ and an integer $M\leq L$, then%
\[
|i-j+p^{L}z|_{p}=|i-j|_{p}\text{ for any }z\in\mathbb{Z}_{p}.
\]
By using this observation, one gets that%
\begin{gather}
\Omega(p^{M}|x|_{p})\ast U(x)=\label{Observation}\\
=\sum_{i\in G_{L}}\left\{  \sum_{j\in G_{L}}U(j)%
{\displaystyle\int\limits_{\mathbb{Z}_{p}}}
\Omega(p^{M}|i-y|_{p})\Omega(p^{L}|y-j|_{p})dy\right\}  \Omega(p^{L}%
|x-i|_{p})\nonumber\\
=\sum_{i\in G_{L}}\left\{  \sum_{j\in G_{L}}U(j)%
{\displaystyle\int\limits_{j+p^{L}\mathbb{Z}_{p}}}
\Omega(p^{M}|i-y|_{p})dy\right\}  \Omega(p^{L}|x-i|_{p})\nonumber\\
=\sum_{i\in G_{L}}\left\{  p^{-L}\sum_{j\in G_{L}}U(j)\Omega(p^{M}%
|i-j|_{p})\right\}  \Omega(p^{L}|x-i|_{p}).\nonumber
\end{gather}
Now, from (\ref{Definition_B})-(\ref{Observation}), we get the following
formula:%
\begin{gather}
(B\ast U)(x)=\label{Definition_C}\\
\sum_{i\in G_{l}}p^{-L}\left(  p^{M_{2}-M_{1}}\sum_{j\in G_{l}}\Omega
(p^{M_{2}}|i-j|_{p})U(j)-\sum_{j\in G_{l}}\Omega(p^{M_{1}}|i-j|_{p}%
)U(j)\right)  \Omega(p^{L}|x-i|_{p}).\nonumber
\end{gather}
We now replace (\ref{Definition_A})-(\ref{Definition_C}) in the equation in
(\ref{Cauchy_1}) and use that $\{\Omega(p^{L}|x-i|)\}_{i\in G_{L}}$ is a basis
of $\mathcal{D}_{L}(\mathbb{Z}_{p})$, to get \ a discretization of
(\ref{Cauchy_1}):%
\begin{equation}
\left\{
\begin{array}
[c]{ll}%
\frac{dX(i,t)}{dt}=-X(i,t)+aY(i,t)+p^{-L}\left(  \mathcal{L}U\right)
(i)+Z(i), & i\in G_{L}\\
X(i,0)=X_{0}(i), &
\end{array}
\right.  \label{CNN_1_discrete}%
\end{equation}
where
\[
Y(i,t)=f\left(  X(i,t)\right)  ,\text{ }i\in G_{L},
\]
and
\begin{equation}
\left(  \mathcal{L}U\right)  (i):=p^{M_{2}-M_{1}}\sum_{j\in G_{l}}%
\Omega(p^{M_{2}}|i-j|_{p})U(j)-\sum_{j\in G_{l}}\Omega(p^{M_{1}}%
|i-j|_{p})U(j),\text{ }i\in G_{L}. \label{Discrete Laplacian}%
\end{equation}

\subsubsection{Graph Laplacians}

Let $G=(V,E)$ be a simple graph with vertices $V$ and edges $E$. Let
$\phi:V\rightarrow\mathbb{R}$ be a function on the graph. The graph Laplacian
$\Delta$ acting $\phi$ is defined as
\[
\left(  \Delta\phi\right)  \left(  v\right)  =%
{\displaystyle\sum\limits_{\substack{w\in V\\dist(w,v)=1}}}
\left[  \phi\left(  v\right)  -\phi\left(  w\right)  \right]  ,
\]
where $dist(w,v)$ is the distance on the graph. Now, let $N\left(  v\right)  $
be a fixed neighborhood of $v$, for instance,%
\[
\mathcal{N}\left(  v\right)  =\left\{  w\in V;dist(w,v)\leq M\right\}  ,
\]
for positive integer $M$, a generalization of operator $\Delta$ is
\begin{equation}
\left(  \Delta_{\mathcal{N}}\phi\right)  \left(  v\right)  =%
{\displaystyle\sum\limits_{w\in\mathcal{N}\left(  v\right)  }}
\left(  \Delta\phi\right)  \left(  w\right)  . \label{A_form}%
\end{equation}
The operator $\mathcal{L}$ has the form (\ref{A_form}). Indeed, the following
formula holds for the operator $\left(  \mathcal{L}U\right)  (i)$:%
\begin{equation}
\left(  \mathcal{L}U\right)  (i)=\sum_{\substack{j\in G_{L}\\|i-j|_{p}\leq
p^{-M_{2}}}}\left[  \sum_{\substack{k\in G_{L}\\p^{-M_{2}}<|j-k|_{p}\leq
p^{-M_{1}}}}\left[  U(j)-U(k)\right]  \right]  \text{, }i\in G_{L}\text{.}
\label{Formula_Graph_Laplacian}%
\end{equation}
In particular, taking $M_{1}=0$, $M_{2}=1$, one gets that%
\[
\sum_{\substack{k\in G_{L}\\p^{-1}<|j-k|_{p}\leq1}}\left[  U(j)-U(k)\right]
=\sum_{\substack{k\in G_{L}\\|j-k|_{p}=1}}\left[  U(j)-U(k)\right]  ,
\]
which is the graph Laplacian on $\ G_{L}=\mathbb{Z}_{p}/p^{L}\mathbb{Z}_{p}$
with the distance induced by $|\cdot|_{p}$.

Finally, we establish formula (\ref{Formula_Graph_Laplacian}). We use that
\[
\#\left\{  k\in G_{L};p^{-M_{2}}<|j-k|_{p}\leq p^{-M_{1}}\right\}
=p^{M_{2}-M_{1}},
\]
since $i=i_{0}+i_{1}p+\ldots+i_{L-1}p^{L-1}$. Now $\left\vert i-j\right\vert
_{p}\leq p^{-M_{2}}$ and $p^{-M_{2}}<|j-k|_{p}\leq p^{-M_{1}}$ imply that%
\[
\left\vert i-k\right\vert _{p}=\max\left\{  \left\vert i-j\right\vert
_{p},|j-k|_{p}\right\}  =|j-k|_{p}\leq p^{-M_{1}},
\]
by ultrametric property of $\left\vert \cdot\right\vert _{p}$. Then
\begin{gather*}
\sum_{\substack{j\in G_{L}\\|i-j|_{p}\leq p^{-M_{2}}}}\left[  \sum
_{\substack{k\in G_{L}\\p^{-M_{2}}<|j-k|_{p}\leq p^{-M_{1}}}}\left[
U(j)-U(k)\right]  \right]  =\\
\sum_{\substack{j\in G_{L}\\|i-j|_{p}\leq p^{-M_{2}}}}\left[  \#\left\{  k\in
G_{L};p^{-M_{2}}<|j-k|_{p}\leq p^{-M_{1}}\right\}  \right]  U(j)\\
-\sum_{\substack{j\in G_{L}\\|i-j|_{p}\leq p^{-M_{2}}}}\text{ \ }%
\sum_{\substack{k\in G_{L}\\p^{-M_{2}}<|j-k|_{p}\leq p^{-M_{1}}}}U(k)=\\
\sum_{\substack{j\in G_{L}\\|i-j|_{p}\leq p^{-M_{2}}}}p^{M_{2}-M_{1}}%
U(j)-\sum_{\substack{k\in G_{L}\\|i-k|_{p}\leq p^{-M_{1}}}}U(k)=\left(
\mathcal{L}U\right)  (i).
\end{gather*}

\section{Numerical Examples}

To construct an edge detector using (\ref{CNN_1_discrete}), it requires an
algorithm for splitting a large image into smaller sub-images. Given an image
$I$ of size $(n,m)$, a prime $p$ and an integer $K$, the algorithm divides
image $I$ into sub-images $I_{r}^{\prime}$ of size $(p^{K},p^{K})$ or less.
Then, we use another algorithm to codify sub-image $I_{r}^{\prime}$ as a test
function $Test(I_{r}^{\prime})$. These algorithms are presented in the
Appendix. We process the test function $Test(I_{r}^{\prime})=U$ using network%
\begin{equation}
\left\{
\begin{array}
[c]{ll}%
\frac{dX(i,t)}{dt}=-X(i,t)+aY(i,t)+\sum_{j=0}^{8}\left\{  U(i)-U(i+j3^{2}%
)\right\}  +z_{0}, & i\in G_{L}\\
X(i,0)=0 & \\
Y(i,t)=f(X(i,t)), &
\end{array}
\right.  \label{CNN_1_discrete_2}%
\end{equation}
with $p=3$, $L=4$, $M_{1}=2$, $M_{2}=4$, and $Z(i)=z_{0}\in\mathbb{R}$, for
$i\in G_{L}$, and rescaling $\left(  \mathcal{L}U\right)  (i)$ as
$3^{4}\left(  \mathcal{L}U\right)  (i)$, for $i\in G_{L}$, to get another test
function $Y(i,t_{0};Test(I_{r}^{\prime}))$ taking values in $\left\{
\pm1\right\}  $. Each test function $Y(i,t_{0};Test(I_{r}^{\prime}))$ is
transformed into an image $I_{r}^{edges}$, at the final step, we concatenate
all the images $I_{r}^{edges}$ to obtain a full image $I^{edges}$, which is
the output \ image. The time $t_{0}$ is chosen on a case-by-case basis so that
the edges are as sharp as possible. See Figures \ref{Figure 1}, \ref{Figure 2}.


\begin{figure}
[h]
\begin{center}
\includegraphics[width=0.5\textwidth]{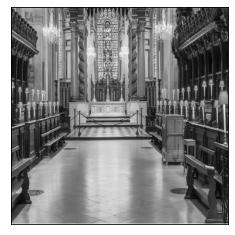}%
\includegraphics[width=0.5\textwidth]{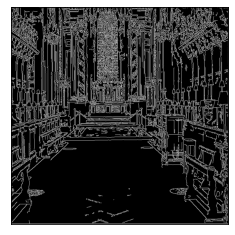}
\caption{Left side, the original image. Right side, edges obtained by using a
Canny edge detector.}%
\label{Figure 1}%
\end{center}
\end{figure}

\begin{figure}
[h]
\begin{center}
\includegraphics[width=0.5\textwidth]{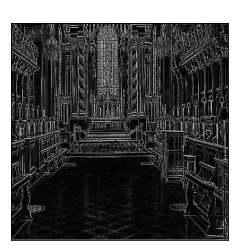}%
\includegraphics[width=0.5\textwidth]{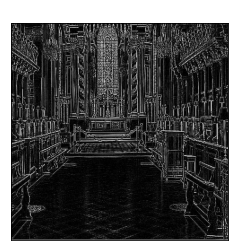}
\caption{Left side, edges obtained by using CNN (\ref{CNN_1_discrete_2}), with
$z_{0}=-1$ and $6$ steps. Right side, edges obtained by using CNN
(\ref{CNN_1_discrete_2}), with $z_{0}=-1$ and $10$ steps.}%
\label{Figure 2}%
\end{center}
\end{figure}

\section{Reaction-diffusion Cellular Neural Networks}

\subsection{The $p$-adic heat equation}

For $\alpha>0$, the Vladimirov-Taibleson operator $\boldsymbol{D}^{\alpha}$ is
defined as
\[%
\begin{array}
[c]{ccc}%
\mathcal{D}(\mathbb{Q}_{p}) & \rightarrow & L^{2}(\mathbb{Q}_{p}%
)\cap\mathcal{C}\left(  \mathbb{Q}_{p}\right) \\
&  & \\
\varphi & \rightarrow & \boldsymbol{D}^{\alpha}\varphi,
\end{array}
\]
where%
\[
\left(  \boldsymbol{D}^{\alpha}\varphi\right)  \left(  x\right)
=\frac{1-p^{\alpha}}{1-p^{-\alpha-1}}%
{\displaystyle\int\limits_{\mathbb{Q}_{p}}}
\frac{\left[  \varphi\left(  x-y\right)  -\varphi\left(  x\right)  \right]
}{\left\vert y\right\vert _{p}^{\alpha+1}}dy.
\]
The $p$-adic analogue of the heat equation is%
\[
\frac{\partial u\left(  x,t\right)  }{\partial t}+a\boldsymbol{D}^{\alpha
}u\left(  x,t\right)  =0\text{, with }a>0\text{.}%
\]
The solution of the Cauchy problem attached to the heat equation with initial
datum $u\left(  x,0\right)  =\varphi\left(  x\right)  \in\mathcal{D}%
(\mathbb{Q}_{p})$ is given by%
\[
u\left(  x,t\right)  =%
{\displaystyle\int\limits_{\mathbb{Q}_{p}}}
Z\left(  x-y,t\right)  \varphi\left(  y\right)  dy,
\]
where $Z\left(  x,t\right)  $ is the $p$\textit{-adic heat kernel} defined as
\begin{equation}
Z\left(  x,t\right)  =%
{\displaystyle\int\limits_{\mathbb{Q}_{p}}}
\chi_{p}\left(  -x\xi\right)  e^{-at\left\vert \xi\right\vert _{p}^{\alpha}%
}d\xi, \label{het-kernel}%
\end{equation}
where $\chi_{p}\left(  -x\xi\right)  $\ is the standard additive character of
the group $\left(  \mathbb{Q}_{p},+\right)  $. The $p$-adic heat kernel is the
transition density function of a Markov stochastic process with space state
$\mathbb{Q}_{p}$, see, e.g., \cite{Kochubei}, \cite{Zuniga-LNM-2016}.

\subsection{The $p$-adic heat equation on the unit ball}

We define the operator $\boldsymbol{D}_{0}^{\alpha}$, $\alpha>0$, by
restricting $\boldsymbol{D}^{\alpha}$ to $\mathcal{D}(\mathbb{Z}_{p})$ and
considering $\left(  \boldsymbol{D}^{\alpha}\varphi\right)  \left(  x\right)
$ only for $x\in\mathbb{Z}_{p}$. The operator $\boldsymbol{D}_{0}^{\alpha}%
$\ satisfies
\[
\boldsymbol{D}_{0}^{\alpha}\varphi(x)=\lambda\varphi(x)+\frac{1-p^{\alpha}%
}{1-p^{-\alpha-1}}%
{\displaystyle\int\limits_{\mathbb{Z}_{p}}}
\frac{\varphi(x-y)-\varphi(x)}{|y|_{p}^{\alpha+1}}dy,
\]
for $\mathbb{\varphi\in}\mathcal{D}(\mathbb{Z}_{p})$, with $\lambda=\frac
{p-1}{p^{\alpha+1}-1}p^{\alpha}$.

Consider the Cauchy problem%
\[
\left\{
\begin{array}
[c]{lll}%
\frac{\partial u\left(  x,t\right)  }{\partial t}+\boldsymbol{D}_{0}^{\alpha
}u\left(  x,t\right)  -\lambda u\left(  x,t\right)  =0\text{, } &
x\in\mathbb{Z}_{p}, & t>0;\\
&  & \\
u\left(  x,0\right)  =\varphi\left(  x\right)  , & x\in\mathbb{Z}_{p}, &
\end{array}
\right.
\]
where $\mathbb{\varphi\in}\mathcal{D}(\mathbb{Z}_{p})$. The solution of this
problem is given by%
\[
u\left(  x,t\right)  =%
{\displaystyle\int\limits_{\mathbb{Z}_{p}}}
Z_{0}(x-y,t)\varphi\left(  y\right)  dy\text{, }x\in\mathbb{Z}_{p}\text{,
}t>0,
\]
where
\begin{align*}
Z_{0}(x,t)  &  :=e^{\lambda t}Z(x,t)+c(t)\text{, }x\in\mathbb{Z}_{p}\text{,
}\\
c(t)  &  :=1-(1-p^{-1})e^{\lambda t}\sum_{n=0}^{\infty}\frac{(-1)^{n}}%
{n!}t^{n}\frac{1}{1-p^{-n\alpha-1}}%
\end{align*}
and $Z(x,t)$ is given (\ref{het-kernel}). The function $Z_{0}(x,t)$ is
non-negative for $x\in\mathbb{Z}_{p}$, $t>0$, and
\[%
{\displaystyle\int\limits_{\mathbb{Z}_{p}}}
Z_{0}(x,t)dx=1,
\]
\cite{Kochubei}. Furthermore, $Z_{0}(x,t)$ is the transition density function
of a Markov process with space state $\mathbb{Z}_{p}$.

The family
\begin{equation}%
\begin{array}
[c]{cccc}%
T_{t}: & L^{1}(\mathbb{Z}_{p}) & \rightarrow & L^{1}(\mathbb{Z}_{p})\\
&  &  & \\
& \phi(x) & \rightarrow & T_{t}\phi(x):=%
{\displaystyle\int\limits_{\mathbb{Z}_{p}}}
Z_{0}(x-y,t)\phi(y)dy,
\end{array}
\label{Khrennikov-Kochubei}%
\end{equation}
is a $C^{0}$-semigroup of contractions with generator $\boldsymbol{D}%
_{0}^{\alpha}-\lambda I$ on $L^{1}(\mathbb{Z}_{p})$, see \cite[Proposition 4,
Proposition 5]{Khrennikov-Kochubei}

\subsection{Reaction-diffusion CNNs}

\begin{definition}
Given $\mu\in\mathbb{R}$, $\alpha>0$, $A$, $B$,$U$, $Z\in\mathcal{C}%
(\mathbb{Z}_{p})$, a $p$-adic reaction-diffusion CNN, denoted as $CNN\left(
\mu,\alpha,A,B,U,Z\right)  $, is the dynamical system given by the following
integro-differential equation:%
\begin{align}
\frac{\partial X(x,t)}{\partial t}  &  =\mu X(x,t)+(\lambda I-\boldsymbol{D}%
_{0}^{\alpha})X(x,t)+%
{\displaystyle\int\limits_{\mathbb{Z}_{p}}}
A(x-y)f(X(y,t))dy\label{RD-CNN}\\
&  +%
{\displaystyle\int\limits_{\mathbb{Z}_{p}}}
B(x-y)U(y)dy+Z(x),\nonumber
\end{align}
where $x\in\mathbb{Z}_{p}$, $t\geq0$. We say that $X(x,t)\in$ $\mathbb{R}$ is
the state of cell $x$ at the time $t$. Function $A$ is the kernel of the
feedback operator, while function $B$ is the kernel of the feedforward
operator. Function $U$ is the input of the CNN, while function $Z$ is the
threshold of the CNN.
\end{definition}

Notice that if $\mu=0$ and $A=B=U=Z=0$, (\ref{RD-CNN}) becomes the $p$-adic
heat equation in the unit ball. Then, in (\ref{RD-CNN}), $(\lambda
I-\boldsymbol{D}_{0}^{\alpha})$ is\ the diffusion term, while the other terms
are the reaction ones, which describe the interaction between $X(x,t)$,
$U(x)$, and $Z(x)$.

\begin{remark}
In this section, we assume that $f$ is an arbitrary Lipschitz function,
$f(0)=0$, i.e., $\left\vert f(s)-f(t)\right\vert \leq L(f)\left\vert
s-t\right\vert $, for $s$, $t\in\mathbb{R}$, where $L(f)$ is a positive constant.
\end{remark}

\begin{lemma}
\label{lemma: lipschiptz condition}Let $A$, $B$, $U$, $Z\in\mathcal{C}%
(\mathbb{Z}_{p}).$

\noindent(i) Set
\begin{equation}
\boldsymbol{H}(g):=%
{\displaystyle\int\limits_{\mathbb{Z}_{p}}}
A(x-y)f\left(  g(y)\right)  dy+%
{\displaystyle\int\limits_{\mathbb{Z}_{p}}}
B(x-y)U(y)dy+Z(x), \label{Definition-Op_H}%
\end{equation}
for $g\in L^{1}(\mathbb{Z}_{p})$. Then $\boldsymbol{H}:L^{1}(\mathbb{Z}%
_{p})\rightarrow L^{1}(\mathbb{Z}_{p})$ is a well-defined \ operator
satisfying%
\[
\Vert\boldsymbol{H}(g)-\boldsymbol{H}(g^{\prime})\Vert_{1}\leq L(f)\Vert
A\Vert_{\infty}\Vert g-g^{\prime}\Vert_{1}\text{, for }g\text{, }g^{\prime}\in
L^{1}(\mathbb{Z}_{p})\text{.}%
\]

\noindent(ii) The restriction of $\boldsymbol{H}$\ to $\mathcal{C}%
(\mathbb{Z}_{p})$ satisfies
\[
\Vert\boldsymbol{H}(g)-\boldsymbol{H}(g^{\prime})\Vert_{\infty}\leq L(f)\Vert
A\Vert_{1}\Vert g-g^{\prime}\Vert_{\infty}\text{, for }g\text{, }g^{\prime}%
\in\boldsymbol{C}(\mathbb{Z}_{p})\text{,}%
\]
so $\boldsymbol{H}:\mathcal{C}(\mathbb{Z}_{p})\rightarrow\mathcal{C}%
(\mathbb{Z}_{p})$ is well-defined operator.
\end{lemma}

\begin{proof}
Take $g,g^{\prime}\in L^{1}(\mathbb{Z}_{p})$, then
\begin{gather*}
\Vert\boldsymbol{H}(g)-\boldsymbol{H}(g^{\prime})\Vert_{1}=\Vert%
{\displaystyle\int\limits_{\mathbb{Z}_{p}}}
A(x-y)\left\{  f\left(  g(y)\right)  -f\left(  g^{\prime}(y)\right)  \right\}
dy\Vert_{1}\\
\leq%
{\displaystyle\int\limits_{\mathbb{Z}_{p}}}
\left\{
{\displaystyle\int\limits_{\mathbb{Z}_{p}}}
\left\vert A(x-y)\right\vert \text{ }\left\vert f\left(  g(y)\right)
-f\left(  g^{\prime}(y)\right)  \right\vert dy\right\}  dx\leq L(f)\Vert
A\Vert_{\infty}%
{\displaystyle\int\limits_{\mathbb{Z}_{p}}}
\left\vert g(y)-g^{\prime}(y)\right\vert dy\\
\leq L(f)\Vert A\Vert_{\infty}\Vert g-g^{\prime}\Vert_{1}.
\end{gather*}
This inequality also proves that $\boldsymbol{H}$ is well-defined. The second
part is established in a similar way.
\end{proof}

\begin{proposition}
\label{Lemma-1}Let $A$, $B$, $U$, $Z\in\mathcal{C}(\mathbb{Z}_{p})$. Take
$X_{0}\in L^{1}(\mathbb{Z}_{p})$ as the initial datum for the Cauchy problem
attached to (\ref{RD-CNN}). Then there exists $\tau=\tau\left(  X_{0}\right)
\in\left(  0,\infty\right]  $ and a unique $X(t)\in\mathcal{C}([0,\tau
],L^{1}(\mathbb{Z}_{p}))$ satisfying
\begin{equation}
\left\{
\begin{array}
[c]{l}%
X(t)=e^{\mu t}T_{t}X_{0}+\int_{0}^{t}e^{\mu(t-s)}T_{t-s}\boldsymbol{H}%
(X(s))ds\\
X(0)=X_{0}.
\end{array}
\right.  \label{Eq:mild solution}%
\end{equation}

\end{proposition}

\begin{proof}
By \cite[Proposition 4]{Khrennikov-Kochubei}, $(\boldsymbol{D}_{0}^{\alpha
}-\lambda I)$ is the generator of a strongly continuous semigroup $\left\{
T_{t}\right\}  _{t\geq0}$ of contraction on $L^{1}(\mathbb{Z}_{p})$. Then
$(\boldsymbol{D}_{0}^{\alpha}-\lambda I)+\mu I$ is the generator of a strongly
continuous semigroup $\left\{  e^{\mu t}T_{t}\right\}  _{t\geq0}$ on
$L^{1}(\mathbb{Z}_{p})$, see \cite[Theorem 4.3-(10)]{Milan}. Since $\left\Vert
e^{\mu t}T_{t}\right\Vert \leq e^{\mu t}$ and $\boldsymbol{H}$ is a Lipschitz
nonlinearity, see Lemma \ref{lemma: lipschiptz condition}-(i), there exits a
unique \ mild solution $X(t)\in\mathcal{C}([0,\tau],L^{1}(\mathbb{Z}_{p}%
))$\ satisfying (\ref{Eq:mild solution}), see, e.g., \cite[Theorem
5.1.2]{Milan}.
\end{proof}

\begin{lemma}
\label{Lemma-2}Let $A$, $B$, $U$, $Z\in\mathcal{C}(\mathbb{Z}_{p})$. Take
$X_{0}\in\mathcal{C}(\mathbb{Z}_{p})$. Then, the integral equation
(\ref{Eq:mild solution}) has unique solution $\mathcal{C}([0,\infty
),\mathcal{C}(\mathbb{Z}_{p}))$.
\end{lemma}

\begin{proof}
It is sufficient to show that (\ref{Eq:mild solution}) has a unique solution
in $\mathcal{C}(\left[  0,T\right]  ,\mathcal{C}(\mathbb{Z}_{p}))$, where
$T>0$ is an arbitrary time horizon. Indeed, if $X_{0}(t)\in\mathcal{C}(\left[
0,T_{0}\right]  ,\mathcal{C}(\mathbb{Z}_{p}))$ and $X_{1}(t)\in\mathcal{C}%
(\left[  0,T_{1}\right]  ,\mathcal{C}(\mathbb{Z}_{p}))$, with $T_{0}\leq
T_{1}$, are mild solutions, then $X_{0}(t)=X_{1}(t)$ for $t\in\left[
0,T_{1}\right]  $, see \cite[Theorem 5.2.3]{Milan}.

We set $\mathcal{Y}:=\mathcal{C}([0,T],\mathcal{C}(\mathbb{Z}_{p}))$, which is
a Banach space with \ norm
\[
\sup_{0\leq t\leq T}\left\Vert Y\left(  t\right)  \right\Vert _{\infty}%
=\sup_{0\leq t\leq T}\left[  \sup_{x\in\mathbb{Z}_{p}}\left\vert
Y(x,t)\right\vert \right]  .
\]
We now set
\[
\boldsymbol{G}g(t):=e^{\mu t}T_{t}X_{0}+\int_{0}^{t}e^{\mu(t-s)}%
T_{t-s}\boldsymbol{H}\left(  g(s)\right)  ds,
\]
for $g(t)\in\mathcal{C}([0,T],\mathcal{C}(\mathbb{Z}_{p}))$. By using that
$Z_{0}(x,t)\in L^{1}(\mathbb{Z}_{p})$, one gets $T_{t}g\in\mathcal{C}%
([0,T],\mathcal{C}(\mathbb{Z}_{p}))$, and by Lemma
\ref{lemma: lipschiptz condition}-(ii), $\boldsymbol{G}:\mathcal{Y}%
\rightarrow\mathcal{Y}$.\ We now set
\[
\boldsymbol{G}^{n}=\underbrace{\underset{n-times}{\boldsymbol{G}%
\circ\boldsymbol{G\circ\cdots\circ G}}}.
\]
We show that for $n$ sufficiently large $\boldsymbol{G}^{n}$ is a contraction.
We first notice that%
\[
\left\Vert \boldsymbol{G}g(t)-\boldsymbol{G}g(t)\right\Vert _{\infty}\leq
L(f)e^{\mu T}\left\Vert A\right\Vert _{1}\Vert g(t)-g^{\prime}(t)\Vert
_{\infty}.
\]

By a well-known argument, see \cite[Proof of Theorem 5.1.2]{Milan}, one gets
that%
\[
\left\Vert \boldsymbol{G}^{n}g(t)-\boldsymbol{G}^{n}g(t)\right\Vert _{\infty
}\leq\frac{(e^{\mu T}L(f)\Vert A\Vert_{1}T)^{n}}{n!}\Vert g(t)-g^{\prime
}(t)\Vert_{\infty},
\]
with $\frac{(e^{\mu T}L(f)\Vert A\Vert_{1}T)^{n}}{n!}<1$, for $n$ sufficiently
large. Therefore $\boldsymbol{G}$ has a unique fixed point $X(t)$ in
$\mathcal{Y}$, see \cite[Theorem 1.1.3]{Milan}.
\end{proof}

\begin{theorem}
\label{Theorem3}Let $X(t)\in\mathcal{C}([0,\infty),\mathcal{C}(\mathbb{Z}%
_{p}))$ be the unique solution of (\ref{Eq:mild solution}), with initial
condition $X_{0}\in\mathcal{C}(\mathbb{Z}_{p})$.Then,
\begin{equation}
\Vert X(t)\Vert_{\infty}\leq e^{\mu t}\Vert X_{0}\Vert_{\infty}+\frac{(e^{\mu
t}-1)}{\mu}\left(  \Vert A\Vert_{1}\Vert f\Vert_{\infty}+\Vert B\Vert_{1}\Vert
U\Vert_{\infty}+\Vert Z\Vert_{\infty}\right)  , \label{Eq-1A}%
\end{equation}
if $\mu\neq0$, otherwise%
\begin{equation}
\Vert X(t)\Vert_{\infty}\leq\Vert X_{0}\Vert_{\infty}+\tau\left(  \Vert
A\Vert_{1}\Vert f\Vert_{\infty}+\Vert B\Vert_{1}\Vert U\Vert_{\infty}+\Vert
Z\Vert_{\infty}\right)  . \label{Eq-1B}%
\end{equation}

\end{theorem}

\begin{proof}
By using that $\left\Vert B\ast U\right\Vert \leq\Vert B\Vert_{1}\Vert
U\Vert_{\infty}$, cf. \cite[Theorem 1.7]{Taibleson}, and Lemma
\ref{lemma: lipschiptz condition}-(ii), we get that%
\[
\Vert\boldsymbol{H}(g)\Vert_{\infty}\leq L(f)\Vert A\Vert_{1}\Vert
g\Vert_{\infty}+\Vert B\Vert_{1}\Vert U\Vert_{\infty}+\Vert Z\Vert_{\infty
}\text{ for }g,\in\boldsymbol{C}(\mathbb{Z}_{p}).
\]
Now, the stated formula follows from (\ref{Eq:mild solution}), by Lemma
\ref{Lemma-2}, \ by using that $\Vert A\Vert_{1}\leq\Vert A\Vert_{\infty}$ and
$\Vert B\Vert_{1}\leq\Vert B\Vert_{\infty}$. The bound (\ref{Eq-1B}) is
established in a similar way.
\end{proof}

\section{Denoising}

In this section, we present a new denoising technique based on certain
reaction-diffusion CNNs. We first consider the initial value problem%
\begin{equation}
\left\{
\begin{array}
[c]{lll}%
\frac{\partial X(x,t)}{\partial t}+\boldsymbol{D}_{0}^{1}X(x,t)-\lambda
X(x,t)=0, & x\in\mathbb{Z}_{p}, & t>0\\
&  & \\
X(x,0)=X_{0}(x), & x\in\mathbb{Z}_{p}, &
\end{array}
\right.  \label{Eq: Vladimirov}%
\end{equation}
where $X_{0}(x)\in$ $[0,1]$ is a grayscale image codified as a test function
supported in the unit ball $\mathbb{Z}_{p}$. The algorithm for this coding is
discussed at the end of this section. The output image $X(x,t)$ is similar to
the one produced by the classical Gaussian filter. See Figure \ref{Figure 3}.


\begin{figure}[h]
\begin{center}
\includegraphics[width=0.5\textwidth]{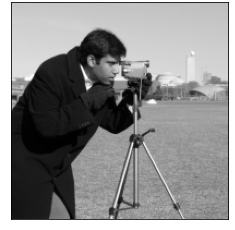}%
\includegraphics[width=0.5\textwidth]{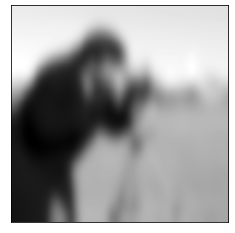}
\caption{On the left side, the original image $X(x,0)$. On the right side
$X(x,3)$.}%
\label{Figure 3}%
\end{center}
\end{figure}

In this article we propose the following reaction-diffusion CNN for denoising
grayscale images polluted with normal additive noise:{\footnotesize
\begin{equation}
\frac{\partial X(x,t)}{\partial t}=3X(x,t)+(\lambda I-\boldsymbol{D}%
_{0}^{\alpha})X(x,t)+3B\ast\left[  X_{0}(x)-f(X(x,t))\right]  , \label{CNN_5}%
\end{equation}
} where $\alpha=0.75$, $f(x)=0.5(|x+1|-|x-1|)$, $B(x)=(\Omega(p^{2}%
|x|_{p})-\Omega(|x|_{p}))$, and $-1\leq X_{0}(x)\leq1$. Notice that we are
using the interval $\left[  -1,1\right]  $ as a grayscale scale. This equation
was found experimentally. Natively, the reaction term $3X(x,t)+3B\ast\left[
X_{0}(x)-f(X(x,t))\right]  $ gives an estimation of the edges of the image,
while the diffusion term $(\lambda I-\boldsymbol{D}_{0}^{\alpha})X(x,t)$
produces a smoothed version of the image.

The processing of an image $X_{0}(x)$ using (\ref{CNN_1_discrete_2}) requires
solving the corresponding Cauchy problem with initial datum $X(x,0)=X_{0}(x)$.
Given an image $I$, i.e., a matrix of size $\left(  n,m\right)  $, and a pixel
$(i,j)$ of $I$ , for the processing of this pixel we use a neighborhood
$I_{i,j}$ centered at this pixel, which is sub-image $I_{i,j}$ of size
$\left(  p^{K},p^{K}\right)  $, where $p^{2K}$ is the number of pixels in the
sub-image $I_{i,j}$. We use small primes, $p=2,3$ to get sub-images of size
$2\times2$ and $3\times3$. The choosing of the prime $p$ is completely
determined by the image size, then, only small primes are required. Now, we
codify the sub-image $I_{i,j}$ a test function $Test(I_{i,j})$ and solve
numerically the Cauchy problem attached to (\ref{CNN_1_discrete_2}) with
initial datum $Test(I_{i,j})$. We pick a time $t_{0}$, on a case by case
basis, and take the test function $X(x,t_{0};I_{i,j})$ as the output of the
network. At the final step, we transform $X(x,t_{0};I_{i,j})$ into an image
$I_{i,j}^{\prime},$ and take the pixel processed image at $(i,j)$ as the
center of $I_{i,j}^{\prime}$. See Figures \ref{Figure 4}, \ref{Figure 5}.


\begin{figure}[h]
\begin{center}
\includegraphics[width=0.5\textwidth]{Figure_3.png}%
\includegraphics[width=0.5\textwidth]{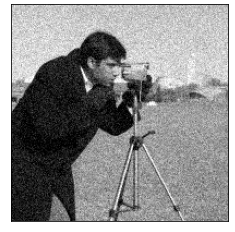}
\caption{Left side, the original image. Right side, the image plus Gaussian
noise, mean zero and variance $0.05$.}%
\label{Figure 4}%
\end{center}
\end{figure}

\begin{figure}[h]
\begin{center}
\includegraphics[width=0.5\textwidth]{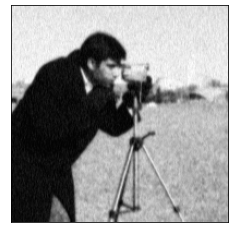}%
\includegraphics[width=0.5\textwidth]{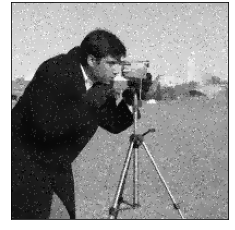}
\caption{Left side, filtered image using Equation \ref{CNN_5}. Right side,
filtered image obtained by using Perona-Malik equation with $\lambda=0.04$,
$\delta_{t}=0.075$, and \ $t=100$ iterations, and $g_{1}(s)$, see
\cite{B-L-Mathm-image}.}%
\label{Figure 5}%
\end{center}
\end{figure}

\section{Appendix: Images and test functions}

In this appendix. we show the existence of a bijective correspondence between
images and test functions. We first show the existence of a bijective
correspondence between finite disjoint unions of balls contained in
$\mathbb{Z}_{p}$, for some prime $p$ with weighted rooted trees of valence
$p$. The connections between clustering, trees and ultrametric spaces are
well-known, see e.g., \cite[Chapter 2]{KKZuniga} and the references therein.
Finally, we show the existence of a bijective correspondence between finite,
regular rooted trees of valence $p$ with images.

\subsection{Finite rooted trees and test functions}

By a \textit{finite rooted tree} $\mathcal{T}$, we mean a finite undirected
graph in which any two vertices are connected by \ exactly one path. The
vertices $V(\mathcal{T})$ of $\mathcal{T}$ are organized in disjoint
\textit{levels}:%
\[
V(\mathcal{T})=%
{\textstyle\bigsqcup\limits_{j=0}^{M}}
Level_{j}\left(  \mathcal{T}\right)  ,
\]
where $Level_{j}\left(  \mathcal{T}\right)  :=Level_{j}=\left\{
v_{j,0},v_{j,1},\ldots,v_{j,k_{j}}\right\}  $, $k_{j}\geq1$, are the vertices
of $\mathcal{T}$ at level $j$. At level $0$ there is exactly one vertex
$v_{0}$, \textit{the root of the tree}. The vertices at the level $1$ are the
\textit{descendants} of the root, which means that there is path $v_{0}$
$\rightarrow v_{1,i}$ for any vertex $v_{1,i}\in Level_{1}$. Inductively, the
vertices at level $j$, $1\leq j\leq M$, are the the descendants of the
vertices at level $j-1$. The vertices at level $M$ do not have descendants.

We denote by $\gamma\left(  v\right)  $, $v\in V(\mathcal{T})$, the number of
edges emanating from $v$. We set
\[
\gamma_{\mathcal{T}}:=\max_{v\in V(\mathcal{T})}\left\{  \gamma\left(
v\right)  \right\}  .
\]
We fix a prime number defined as $p_{\mathcal{T}}:=\min_{p}\left\{
\gamma_{\mathcal{T}}\leq p;\text{ }p\text{ prime}\right\}  $. For the sake of
simplicity we use $p:=p_{\mathcal{T}}$. Given any vertex $v_{j,i_{j}}\in
Level_{j}$, $1\leq j\leq M$, there is exactly one path connecting $v_{j,i_{j}%
}$ with $v_{0}$:%
\begin{equation}
v_{0}\rightarrow v_{1,i_{1}}\rightarrow\ldots\rightarrow v_{j-1,i_{j-1}%
}\rightarrow v_{j,i_{j}}. \label{Path}%
\end{equation}
We attach to $v_{j,i_{j}}$ the $p$-adic integer%
\begin{equation}
I_{v_{j,i_{j}}}:=i_{1}+i_{2}p+\ldots+i_{j-1}p^{j-2}+i_{j}p^{j-1}\text{,}
\label{P-adic_Number}%
\end{equation}
where the digits $i_{k}$ belong to $\left\{  0,1,\ldots,p-1\right\}  $. Then,
there is a bijection between the vertices of $\mathcal{T}$ and the $p$-adic
integers of form (\ref{P-adic_Number}). Given a vertex $v$ at level $L_{v}$
denote the corresponding $p$-adic number as
\begin{equation}
I_{v}=i_{1}+i_{2}p+\ldots+i_{L_{v}-1}p^{L_{v}-1}\text{, \ }L_{v}\leq M.
\label{dato_1}%
\end{equation}

Now we attach to $\mathcal{T}$ the following family of balls:%
\[
B(\mathcal{T}):=\left\{  I_{v}+p^{L_{v}}\mathbb{Z}_{p}\text{, }v\in
V(\mathcal{T})\smallsetminus\left\{  v_{0}\right\}  \right\}
{\textstyle\bigsqcup}
\text{ }\left\{  \mathbb{Z}_{p}\right\}  \text{,}%
\]
where the unit ball $\mathbb{Z}_{p}$ correspond to the case $v=v_{0}$. The
tree $\mathcal{T}$ and the collection of balls $B(\mathcal{T})$ are equivalent
data. Indeed, given a finite collection $B$ of balls contained in
$\mathbb{Z}_{p}$ such that $\mathbb{Z}_{p}\in B$, there is a finite rooted
tree $\mathcal{T}$ that represents the partial order induced by $\subseteq$ in
$B$.

We say that a vertex $v$ is a leaf of $\mathcal{T}$ if $v$ does not have
descendants. In particular all the vertices in $Level_{M}$ are leaves. We
denote by $Leaf(\mathcal{T})$ the set of all leaves of $\mathcal{T}$. Finally,
we attach to $\mathcal{T}$ the open compact subset%
\begin{equation}
\mathcal{K}(\mathcal{T})=\bigsqcup\limits_{v\in Leaf(\mathcal{T})}\text{
}\left(  I_{v}+p^{L_{v}}\mathbb{Z}_{p}\right)  . \label{dato_2}%
\end{equation}
Now, given a finite disjoint union of balls of the form $%
{\textstyle\bigsqcup\nolimits_{v\in\mathcal{G}}}
\left(  I_{v}+p^{L_{v}}\mathbb{Z}_{p}\right)  $, there is a unique tree
$\mathcal{T}$ \ having $\left\{  I_{v};v\in\mathcal{G}\right\}  $ as\ a set of
leaves. The other vertices correspond to truncations of the numbers $I_{v}$s.
And given a tree $\mathcal{T}$, (\ref{dato_2}) attaches a unique finite
disjoint union of balls to $\mathcal{T}$.

We define a \textit{weighted tree} as a pair $\left(  \mathcal{T},w\right)  $,
where $w:$ $Leaf(\mathcal{T})\rightarrow\mathbb{R}_{+}:=\left\{
x\in\mathbb{R};x\geq0\right\}  $. \ We denote by $\Omega\left(  p^{L_{v}%
}\left\vert x-I_{v}\right\vert _{p}\right)  $ the characteristic function of
the ball $\left(  I_{v}+p^{L_{v}}\mathbb{Z}_{p}\right)  $. Given a test
function from $\mathcal{D}_{L_{v}}(\mathbb{Z}_{p})$ of the form
\begin{equation}
\Phi\left(  x\right)  =%
{\textstyle\sum\limits_{v\in\mathcal{G}}}
c_{v}\Omega\left(  p^{L_{v}}\left\vert x-I_{v}\right\vert _{p}\right)  \text{,
\ }x\in\mathbb{Z}_{p}, \label{dato_3}%
\end{equation}
we attach to it the unique weighted tree with leaves $\left\{  I_{v}%
;v\in\mathcal{G}\right\}  $ and weights $v\rightarrow c_{v}$, for
$v\in\mathcal{G}$. Conversely, given a weighted tree $\left(  \mathcal{T}%
,w\right)  $, with leaves $\mathcal{G}=\left\{  I_{v};v\in Leaf(\mathcal{T}%
)\right\}  $, and $w(v)=c_{v}$ for $v\in\mathcal{G}$, (\ref{dato_3}) defines a
unique test function $\Phi\left(  x\right)  $ from $\mathcal{D}_{L_{v}%
}(\mathbb{Z}_{p})$.

\subsection{Images and finite rooted trees}

In the numerical simulations, we use an algorithm for coding an image as a
finite, weighted, regular, rooted tree of valence $p$, where $p$ is a prime
number. The input is an image $I$, a $(n,m)$ matrix, and a prime number $p$
satisfying $p\leq m,n$. The output is a finite, weighted, regular tree
$Tree(I)$. We use two functions: the function $d_{H}$ divides an image into
$p$ horizontal sub-images, and the function $d_{V}$ divides an image into $p$
vertical sub-images. The tree has at most $L:=\lfloor\log_{p}(nm)\rfloor$
\ levels. The level zero contains just the root of the tree. Each vertex of
the tree corresponds to a sub-image $I^{\prime}$ of $I$, and the descendants
of this vertex, in the next level, are sub-images of $I^{\prime}$ obtained by
using the function $d_{H}$ or $d_{V}$.

The tree $Tree(I)$ corresponding to an image $I$\ is construct recursively as follows:

\begin{enumerate}
\item Level $0$: there is one vertex, the root of the tree which corresponds
to $I$.

\item Level $2l+1$: the descendants of a vertex $I^{\prime}$ at the level $2l$
correspond to the elements of $d_{H}(I^{\prime})$.

\item Level $2l$: the descendants of a vertex $I^{\prime}$ at the level $2l-1$
correspond to the elements of $d_{V}(I^{\prime})$.

\item Level $L$: \ all the vertices (leaves) at the level $L$\ are pixels. The
grayscale intensity of each pixel gives a the weight of the corresponding leaf.
\end{enumerate}

We now define the operator $d_{V}$. Let \ $m_{0},r_{0}$ nonnegative integers
such that $m=pm_{0}+r_{0}$. If $m_{0}\neq0$, we define
\begin{align*}
I_{s}  &  =\left[  I_{i,j}\right]  _{\substack{0\leq i<n\\sm_{0}\leq j\leq
m_{0}(s+1)}}\text{ for }s=0,\ldots,r_{0},\\
I_{s}  &  =\left[  I_{i,j}\right]  _{\substack{0\leq i<n\\m_{0}s+r_{0}\leq
j<m_{0}s}}\text{ for }s=r_{0}+1,\ldots,p-1.
\end{align*}
and
\[
d_{V}(I)=[I_{s}]_{s=0,\ldots,p-1}.
\]
If $m_{0}=0$, we define
\[
I_{s}=\left[  I_{i,j}\right]  _{\substack{0\leq i<n\\j=s}}\text{ for
}s=0,\ldots,r_{0}\text{, and }d_{V}(I)=[I_{s}]_{s=0,\ldots,p-1}.
\]
Thus the operator $d_{V}$ divides the image $I$ into $p$ vertical sub-images.

We now define operators $d_{H}$. Let $n_{0},q_{0}$ be non-negative integers
satisfying $n=(p-1)n_{0}+q_{0}$. If $n_{0}\neq0$. \ We define
\begin{align*}
I_{s}  &  =\left[  I_{i,j}\right]  _{\substack{sn_{0}\leq i\leq n_{0}%
(s+1)\\0\leq j<m}}\text{ for }s=0,\ldots,q_{0},\\
I_{s}  &  =\left[  I_{i,j}\right]  _{\substack{n_{0}s+q_{0}\leq i<n_{0}%
s\\0\leq j<m}}\text{ for }s=q_{0}+1,\ldots,p-1,
\end{align*}
and%
\[
d_{H}(I)=[I_{s}]_{s=0,\ldots,p-1}\text{.}%
\]
If $n_{0}=0$, we define
\[
I_{s}=(I_{i,j})_{i=s;0\leq j<m}\text{ for }s=0,\ldots,q_{0}\text{, and }%
d_{H}(I)=[I_{s}]_{s=0,\ldots,p-1}.
\]

Thus the operator $d_{H}$ divides the image $I$ into $p$ horizontal sub-images.

Consequently, the correspondence between images and weighted, finite, regular,
rooted trees of valence $p$, is a bijection. Figure \ref{Figure 6} shows the
correspondence between images and test functions.

\begin{figure}[h]
\begin{center}
\includegraphics[width=0.4\textwidth]{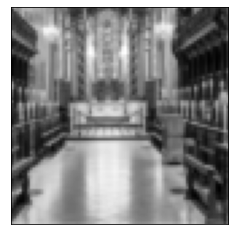}
\includegraphics[width=0.5\textwidth]{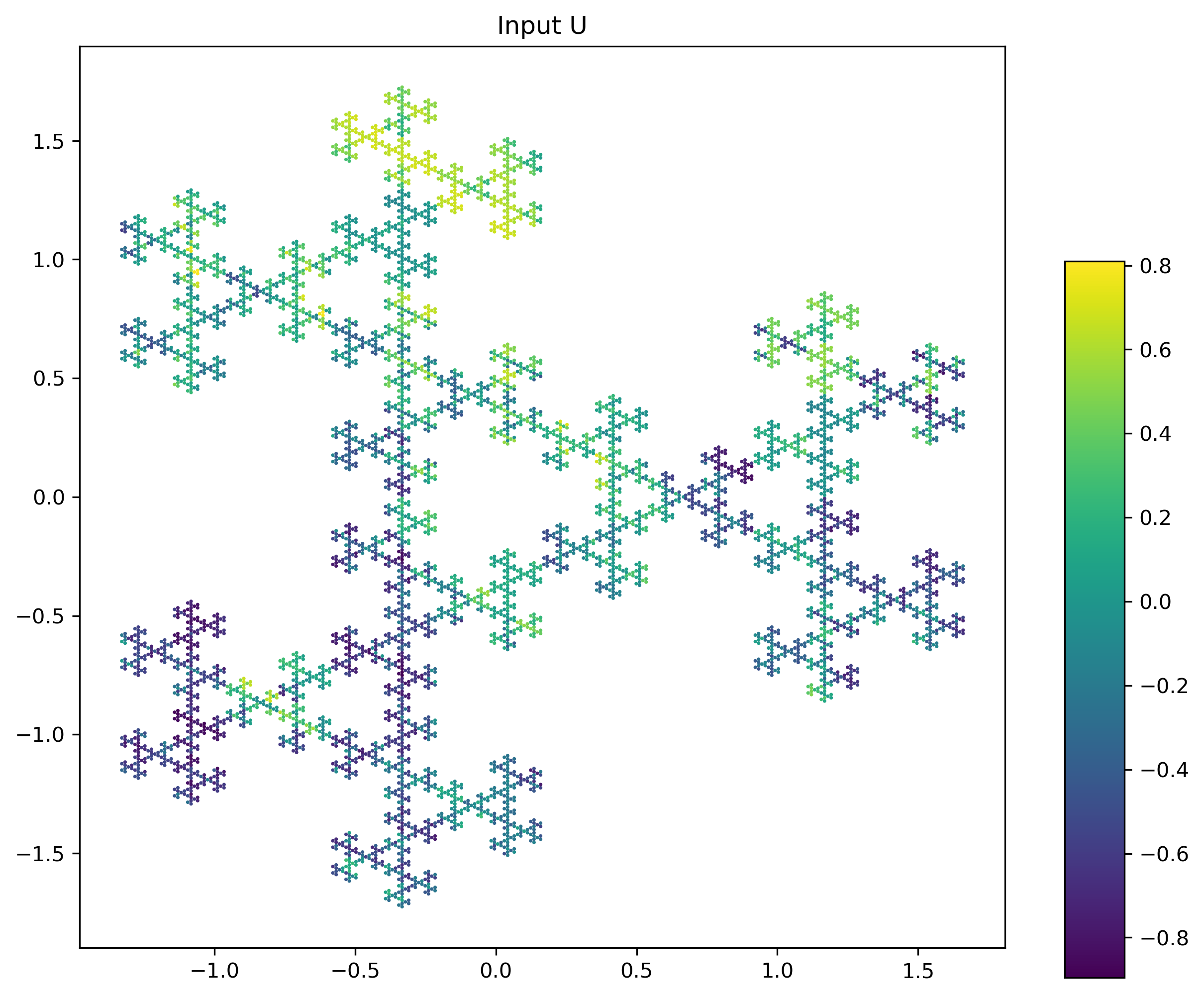}
\caption{Left side, original image $81\times81$. Right side, the
representation of the image as atest function. We use $p=3$, $L=8$.}%
\label{Figure 6}%
\end{center}
\end{figure}

\end{document}